\documentclass[runningheads,envcountsame,a4paper]{llncs}

\usepackage[utf8]{inputenc}

\usepackage{latexsym}
\usepackage{amsbsy}
\usepackage{amssymb} 
\usepackage{amsmath} 
\usepackage{mathbbol} 
\usepackage{stmaryrd}  

\usepackage{url}

\usepackage{comment}
\excludecomment{RR}\includecomment{ABS} 

\usepackage{float}
\floatstyle{boxed}
\newfloat{mafigure}{thp}{pipo}
\floatname{mafigure}{Figure}

\usepackage{pgf}
\usepackage{tikz}
\usetikzlibrary{trees}
\usetikzlibrary{matrix}

\usepackage{xspace}

\def\<#1>{\langle #1 \rangle}

\newcommand{\R}{\mathcal{R}}
\newcommand{\T}{\mathcal{T}}
\renewcommand{\H}{\mathcal{H}}

\newcommand{\X}{\mathcal{X}}

\newcommand{\F}{\Sigma}
\newcommand{\Q}{\ensuremath{Q}}

\newcommand{\A}{\mathcal{A}}
\newcommand{\B}{\mathcal{B}}
\newcommand{\C}{\mathcal{C}}

\newcommand{\var}{\mathit{var}}

\newcommand{\ptrs}[2]{{#1}{/}{#2}}

\newcommand{\emptyhedge}{\varepsilon}

\newcommand{\TRS}{\textup{HRS}\xspace}
\newcommand{\PTRS}{\textup{P}\TRS}
\newcommand{\uPTRS}{\textup{u}\PTRS}

\newcommand{\HA}{\textsf{HA}\xspace}
\newcommand{\CFHA}{$\textsf{CFHA}$\xspace}
\newcommand{\sCFHA}{$\textsf{CF}^2\textsf{HA}$\xspace}
\newcommand{\final}{\mathsf{f}}

\newcommand{\pre}{\mathit{pre}}
\newcommand{\post}{\mathit{post}}

\newcommand{\dom}{\mathit{dom}}

\newcommand{\hor}{\mathsf{h}}
\newcommand{\ver}{\mathsf{v}}
\newcommand{\stack}[2]{\ensuremath{{#2}^{#1}}}
\newcommand{\ustack}[2]{\ensuremath{{#2}_{#1}}}
\newcommand{\state}[4]{\ensuremath{{#1}^{#2}_{\ifthenelse{\equal{#3}{}}{}{{#3},}{#4}}}}

\newcommand{\nt}[4]{\ensuremath{\langle{#1}^{#2}_{\ifthenelse{\equal{#3}{}}{}{{#3},}{#4}}\rangle}}

\newcommand{\initial}[2]{\stack{#2}{#1}}

\newcommand{\CHILD}{\mathsf{ac}}
\newcommand{\SIB}{\mathsf{as}}
\newcommand{\PAR}{\mathsf{ap}}
\newcommand{\REN}{\mathsf{ren}}
\newcommand{\RPL}{\mathsf{rpl}}
\newcommand{\DEL}{\mathsf{del}}

\newcommand{\Pin}{\ensuremath{P_{\mathsf{in}}}}

\newcommand{\hospital}{\mathsf{hospital}}
\newcommand{\patient}{\mathsf{patient}}

\def\frew#1#2#3#4#5#6#7#8{
\setbox0=\hbox{$#6 #7 #1 #8$}%
\setbox1=\hbox{$#6 #7 #2 #8$}%
\ifdim \wd0>\wd1 \rlap{\rlap{\hbox to \wd0{#5}}%
                            {\hbox to\wd0{\hfil\lower #3\box1\relax\hfil}}}{\raise #4\box0}%
\else \rlap{\rlap{\hbox to \wd1{#5}}{\hbox to\wd1{\hfil\raise #4\box0\relax\hfil}}}{\lower #3\box1}%
\fi
}

\def\fstep#1#2#3#4#5{\mathchoice{\frew{#1}{#2}{1.10ex}{1.20ex}{#5}{\scriptstyle}{#3}{#4}}%
                                {\frew{#1}{#2}{1.02ex}{1.20ex}{#5}{\scriptstyle}{#3}{#4}}%
                                {\frew{#1}{#2}{0.51ex}{0.82ex}{#5}{\scriptscriptstyle}{#3}{#4}}%
                                {\frew{#1}{#2}{0.51ex}{0.69ex}{#5}{\scriptscriptstyle}{#3}{#4}}}

\newcommand{\lrstep}[2]{\mathrel{\fstep{#1}{#2}{\;\>}{\>\>\;}{\rightarrowfill}}}

\usepackage{ifthen}

\title{Rewrite Closure and CF Hedge Automata\thanks{%
This work has been partly supported by the ANR ContINT grant INEDIT (2012-15).}}
\author{Florent Jacquemard \inst{2} \and Michael Rusinowitch \inst{1}}

\institute{%
INRIA Nancy--Grand Est \& LORIA UMR -- \email{rusi@loria.fr}\\
615 rue du Jardin Botanique, 54602 Villers-les-Nancy, France.
\and
INRIA Paris--Rocquencourt \& Ircam UMR --
\email{florent.jacquemard@inria.fr}\\
1 place Igor Stravinsky, 75004 Paris, France.
}

\titlerunning{Rewrite Closure and CF Hedge Automata}
\authorrunning{F. Jacquemard and M. Rusinowitch}

\toctitle{Rewrite Closure and CF Hedge Automata}
\tocauthor{Florent~Jacquemard and Michael~Rusinowitch}

\date{\today}

\begin{document}

\maketitle

\setcounter{footnote}{0}

\begin{abstract}
We  introduce an extension of 
hedge automata
called bidimensional context-free hedge automata. 
The class of unranked ordered tree languages they recognize
is shown to be preserved by rewrite closure with inverse-monadic rules. 
We also extend the parameterized rewriting rules 
used for modeling the W3C XQuery Update Facility in previous works,  by the possibility to insert a new parent node above a given node.
We show that the rewrite closure of hedge automata languages with these  
extended rewriting systems are  context-free hedge languages. 
%
\end{abstract}

\section*{Introduction}
Hedge Automata (HA) are  extensions of tree automata to 
manipulate unranked ordered trees. They appeared as a natural tool to support 
document validation since  the number of children of a node is not fixed in XML documents 
and the structural information (type)  of an XML document can be specified by an HA. 
 
A central problem in XML document processing is  
\emph{static typechecking}. This problem 
amounts to verifying at compile time 
that every output XML document which is the result of a specified query or transformation 
applied to an input document with a valid input type  has a valid output type.
However  for  transformation languages such as the one provided by 
XQuery Update Facility (XQUF), the output type of 
(iterated) applications of update primitives are not easy to predict. 
Another important issue for  XML data processing 
is the specification and enforcement of access policies. 
A large amount of work has been devoted to secure XML querying. 
But most of the work focuses on read-only rights, and very few have considered update rights 
for a model based on XQUF operations~\cite{FundulakiManeth07,Bravo08,JacquemardRusinowitch10ppdp}. 
These works have considered the sensitive problem  of access control policy inconsistency,  
that is, {\em whether a  forbidden operation can be simulated through a sequence of allowed operations}. 
For instance \cite{JacquemardRusinowitch10ppdp} presents a  $\hospital$ database  example  
where it is forbidden to rename a $\patient$ name in a medical file but the same effect can be obtained 
by deleting this file and inserting a new one. This example illustrates a so-called {\em  local inconsistency} problem 
and its detection can be reduced 
to checking the emptiness of a HA language. 

In formal verification of infinite state systems several regular model checking approaches 
represent sets of configurations by regular languages, 
transitions by  rewrite rules and (approximations of)  reachable  configurations
as  rewrite closure of regular  languages 
see e.g.~\cite{FeuilladeGT04,BouajjaniT05}.
Regular model checking~\cite{Bouajjani00regular} is extended from tree to hedge rewriting and hedge automata 
in~\cite{Touili07}, which gives a procedure to compute  reachability sets \emph{approximations}.
Here we compute exact reachability sets when the configuration sets are 
represented by  context-free hedge automata, hence beyond the regular (HA) ones. 
These results are interesting for automated verification where reachability sets are not always regular.

To summarize,  several XML validation or infinite-state verification 
problems would  benefit from procedures to compute rewrite-closure of hedge languages. 
We also need decidable formalisms beyond regular tree languages to capture rewrite closures. 

\noindent\textit{Contributions.}
In \cite{JacquemardRusinowitch10ppdp} we have proposed a model for XML update primitives of  XQUF 
as parameterized rewriting rules of the form: 
"insert an unranked tree from a regular tree language $L$ as the first child of a node labeled by $a$".
For these rules, we give  type inference algorithms, 
considering types defined by several classes of unranked tree automata.
In particular we have considered context-free hedge automata (CFHA, e.g. \cite{JR-rta2008}),  
a more general class than regular hedge automata 
and obtained by requiring that the sequences of sibling states  under a node to be 
in a context-free language. 
In this submission we first introduce a non-trivial extension of context-free hedge languages 
defined by what we call bidimensional context-free hedge automata (Section~\ref{sec:bidim}). 
This class is more expressive as shown by examples. 
The class is  also shown to be preserved by rewrite closure when applying 
inverse-monadic rules that are more general than the rules that were considered in~\cite{JR-rta2008}(Section \ref{sec:TRS}).

Then we extend the parameterized rewriting  rules 
used for modeling XQUF in~\cite{JacquemardRusinowitch10ppdp}
by the possibility to insert a new parent node above a given node.
We show in Section~\ref{sec:PTRS} how to compute the rewrite closure of HA languages with these  
extended rewriting systems. Although the obtained results are more general than \cite{JacquemardRusinowitch10ppdp} 
the proofs are somewhat simpler thanks to a new uniform representation 
of vertical and horizontal steps of CFHA. 
\begin{ABS}
A full version 
is available at~\cite{versionlongue}.
\end{ABS}

\noindent\textit{Related work.}
\cite{Delzanno12gandalf} presents a  static analysis of  XML document adaptations, 
expressed as sequences of XQUF primitives. The authors also use an automatic inference 
method for deriving  the type, expressed as a HA, of a sequence of document updates. 
The type is computed starting from the original schema  and from the  XQuery Updates  
formulated as rewriting rules as in~\cite{JacquemardRusinowitch10ppdp}.
However differently from our case the updates are applied in  parallel in one shot.


\section{Preliminaries}
We consider a finite alphabet $\F$ and an infinite set of variables $\X$.
The symbols of $\F$ are generally denoted $a,b,c\ldots$ and the variables
$x$, $y$\ldots
The sets of \emph{hedges} and \emph{trees} over $\F$ and $\X$, 
respectively denoted $\H(\F, \X)$ and $\T(\F, \X)$,
are defined recursively as the smallest sets such that:
every $x \in \X$ is a tree,
if $t_1,\ldots,t_n$ is a finite sequence of trees (possibly empty),
then $t_1 \ldots t_n$ is a hedge and
if $h$ is a hedge and $a\in \F$, then $a(h)$ is a tree.
The empty hedge (case $n \geq 0$ above) is denoted $\varepsilon$ and 
the tree $a(\emptyhedge)$ will be simply denoted by $a$.
We use the operator $.$ to denote the concatenation of hedges.
%
A root (resp. leaf) of a hedge $h = (t_1\ldots t_n)$ 
is a root node (resp. leaf node, i.e. node without child) 
of one of the trees $t_1,..., t_n$.
The root node of $a(h)$ is called the \emph{parent} of every root of $h$
and every root of $h$ is called a \emph{child} of the root of $a(h)$.


We will sometimes consider a tree as a hedge of length one, 
\textit{i.e.} consider that $\T(\F, \X) \subset \H(\F, \X)$.
The sets of ground trees (trees without variables) and ground hedges
are respectively denoted $\T(\F)$ and $\H(\F)$.
The set of variables occurring in a hedge $h \in \H(\F, \X)$ is denoted $\var(h)$.
A hedge $h \in \H(\F, \X)$ is called \emph{linear} if every variable of 
$\var(h)$ occurs once in $h$.
%
A \emph{substitution} $\sigma$ is a mapping of finite domain from $\X$ into $\H(\F, \X)$%
\begin{ABS}
, whose application (written with postfix notation) is extended homomorphically to $\H(\F, \X)$.
\end{ABS}
\begin{RR}
.
The application of a substitution $\sigma$ to terms and hedges 
(written with postfix notation)
is defined recursively by
$x \sigma := \sigma(x)$ when $x \in \dom(\sigma)$,
$y \sigma := y$ when $y \in \X \setminus \dom(\sigma)$,
$(t_1 \ldots t_n) \sigma := (t_1\sigma \ldots t_n \sigma)$  for $n \geq 0$,
and $a(h) \sigma := a(h\sigma)$.
\end{RR}
The set $\C(\F)$ of \emph{contexts} over $\F$ 
contains the linear hedges of $\H\bigl(\F, \{ x \}\bigr)$.
The application of a context $C\in \C(\F)$ to a hedge $h \in \H(\F,\X)$ 
is defined by $C[h] := C \{ x \mapsto h \}$.
\begin{RR}
It consists in inserting $h$ in $C$ in place of the node labelled by $x$.
Sometimes, we write $h[s]$ in order to emphasize that $s$ is 
a subhedge (or subtree) of $h$.
\end{RR}

A \emph{hedge rewriting system} (\textbf{HRS}) $\R$ over a finite unranked alphabet $\F$ 
is a set of \emph{rewrite rules} of the form $\ell \to r$ 
where $\ell \in \H(\F,\X) \setminus \X$ and $r \in \H(\F,\X)$;
$\ell$ and $r$ are respectively called left- and right-hand-side 
(\emph{lhs} and \emph{rhs}) of the rule.
Note that we do not assume the cardinality of $\R$ to be finite.
A \TRS is called ground, resp. linear,
if all its \textit{lhs} and \textit{rhs} of rules are ground, resp. linear.

The rewrite relation $\lrstep{}{\R}$ of a \TRS $\R$ is
the smallest binary relation on $\H(\F, \X)$ containing $\R$
and closed by application of substitutions and contexts.
In other words, $h \lrstep{}{\R} h'$, 
iff there exists a context $C$, 
a rule $\ell \to r$ in $\R$ and a substitution $\sigma$
such that $h = C[ \ell \sigma ]$ and $h' = C[ r\sigma ]$.
The reflexive and transitive closure of $\lrstep{}{\R}$ is denoted $\lrstep{*}{\R}$.
\noindent 
Given $L \subseteq \H(\F,\X)$ and a \TRS $\R$, 
we define the \emph{rewrite closure} of $L$ under $\R$ as $\post_{\R}^*(L) := 
\{ h' \in \H(\F,\X) \mid \exists h \in L, h \lrstep{*}{{\R}} h' \}$%
\begin{ABS}
.
\end{ABS}
\begin{RR}
 and the \emph{backward rewrite closure} as $\pre_{\R}^*(L) := 
 \{ h \in \H(\F,\X) \mid \exists h' \in L, h \lrstep{*}{{\R}} h' \}$.
\end{RR}

\begin{example} \label{ex:TRS}
Let us consider the following rewrite rules 
\[
\R = \{ p_0(x) \to a.p_1(x), p_1(x) \to p_2(x).c, p_2(x) \to p_0(b(x)), p_2(x) \to b(x) \}.
\]
Starting from $p_0 = p_0(\emptyhedge)$, we have the following rewrite sequence
\begin{ABS}
\(
\end{ABS}
\begin{RR}
\[
\end{RR}
p_0 \to a.p_1 \to a.p_2.c \to a.p_0(b).c \to a.a.p_1(b).c \to a.a.p_2(b).c.c \to a.a.p_0(b(b)).c.c \to \ldots 
\begin{ABS}
\)
\end{ABS}
\begin{RR}
\]
\end{RR}
\begin{ABS}
The
\end{ABS}
\begin{RR}
We can observe that the 
\end{RR}
trees of the rewrite closure of $p_0$ under $\R$ which do not contain the symbols
$p_0$, $p_1$, $p_2$ is the set of T-patterns of the form 
$a\ldots a.b(\ldots b(b)).c\ldots c$ with the same number of $a$, $b$ and~$c$.
\end{example}


\section{Bidimensional Context-Free Hedge Automata}
\label{sec:bidim}
A \emph{bidimensional context-free hedge automaton} ($\textbf{CF}^2\textbf{HA}$)
is a tuple 
$\A$ =\\ $\< \F, Q, Q^\final, \Delta>$ 
where  $\F$ is a finite unranked alphabet, 
$Q$ is a finite set of states disjoint from $\F$,
$Q^\final \subseteq Q$ is a set of final states, 
and $\Delta $ is a set of rewrite rules of one of the following form, 
where $p_1, \ldots, p_n \in Q \cup \F$, $q\in Q$ and $n \geq 0$
\[
\begin{array}{rclcl}
p_1(x_1) \ldots p_n(x_n) & \to & q(x_1\ldots x_n) & & \mbox{called \emph{horizontal} transitions,}\\
p_1\bigl(p_2(x)\bigr) & \to & q(x) & & \mbox{called \emph{vertical} transitions.}
\end{array}
\]
The move relation $\lrstep{}{\A}$ between ground hedges of $\H(\F \cup Q)$ 
is defined as the rewrite relation defined by $\Delta$.
The language of a \sCFHA $\A$ in one of its states $q$,
denoted by $L(\A, q)$, 
is the set of ground hedges $h \in \H(\F)$ 
such that $h \lrstep{*}{\A} q$
(we recall that $q$ stands for $q(\varepsilon)$).  
A hedge is accepted by ${\A}$ if there exists $q \in  Q^\final$ such that $h \in L(\A, q)$. 
The language of $\A$, denoted by  $L(\A)$ is the set of hedges accepted by ${\A}$. 
%
\begin{RR}

Note that it is not a limitation in expressiveness to consider only the cases $n \leq 2$ for horizontal transitions.
The case $n = 1$ corresponds to a simple node relabeling rule.  
The case $n = 0$ corresponds to a transition $\emptyhedge \to q$ from the empty hedge.
We can assume \textit{wlog} a unique state $q_\emptyhedge$ such that there is a transition
$\emptyhedge \to q_\emptyhedge$ and that $q_\emptyhedge$ does not occur in \textit{lhs} of horizontal transitions.
Moreover, it is possible to force one variable $x_i$ in an horizontal transition as above to be~$\emptyhedge$.
Say for instance that we want to force $x_1 = \emptyhedge$.
We use a copy $p_1^\emptyhedge$ of the symbol $p_1$, 
a new transition 
$p_1(q_\emptyhedge(x)) \to p_1^\emptyhedge(x)$,
where $q_\emptyhedge$ is as above,
and replace the transition  $p_1(x_1) \ldots p_n(x_n) \to q(x_1\ldots x_n)$ by
$p_1^\emptyhedge(x_1).p_2(x_2) \ldots p_n(x_n) \to q(x_1\ldots x_n)$.
We can apply the same principle to vertical transitions in order to force $x = \emptyhedge$.
\end{RR}
\begin{RR}
Therefore, we shall
\end{RR}
\begin{ABS}
We shall 
\end{ABS}
also consider below the following kind of transitions,
which have the same  expressiveness as \sCFHA.
\[
\begin{array}{rcl}
p_1(\delta_1) \ldots p_n(\delta_n) & \to & q(\delta_1\ldots \delta_n)\\
p_1(p_2(\delta_1)) & \to & q(\delta_1)
\end{array}
\quad
\begin{array}{l}
n > 0\\
\mathrm{every}\; \delta_i\; \mathrm{is~either~a~variable}\; x_i\; \mathrm{or}\; \emptyhedge
\end{array}
\]
\begin{RR}
For instance, $p_1 . p_2(x_2)\ldots p_n(x_n) \to q(x_2\ldots x_n)$ is equivalent to the above horizontal transition.
\end{RR}

\begin{example} \label{ex:T-pattern}
The language of T-patterns over $\F = \{ a, b, c\}$, see Example~\ref{ex:TRS},
is recognized by 
$\< \F, \{ q_0, q_1, q_2 \}, \{ p_0 \}, \Delta>$
with 
$\Delta = \{ 
b(x_1) \to q_0(x_1),\;
a.q_0(x_2) \to q_1(x_2),\; 
q_1(x_1).c \to q_2(x_1),\;
q_2(b(x)) \to q_0(x) \}$.
\end{example}

\subsection{Related Models}

The \sCFHA capture the expressiveness of two models of automata on unranked trees: 
the hedge automaton~\cite{Murata00}
and the lesser known extension of \cite{OhsakiST03} that we call \CFHA.
A \emph{hedge automaton} (\textbf{HA}),
resp. \emph{context-free hedge automaton} (\textbf{CFHA}) 
is a tuple $\A = \<\F, Q, Q^\final, \Delta>$ 
where  $\F$, $Q$ and $Q^\final$ are as above, 
and the transitions of $\Delta$ have the form $a(L) \to q$ 
where $a \in \F$, $q\in Q$ and $L \subseteq Q^*$ is a regular word language 
(resp. a context-free word language).
The language of hedges accepted is defined as for \sCFHA,
using the rewrite relation of $\Delta$.

%
The \CFHA languages form a strict subclass of \sCFHA languages.
Indeed
every \CFHA can be presented as a \sCFHA with variable-free transitions of the form 
\[
p_1 \ldots p_n \to q \quad
a(q_1) \to q_2 \quad
\mathrm{where}\; a \in \F\; \mathrm{and}\; q_1, q_2\; \mathrm{are~states}.
\]
It can be shown that the set of T-patterns of Example~\ref{ex:T-pattern}
is not a \CFHA language, using a pumping argument on the paths labeled by $b$.

The \HA languages, also called \emph{regular} languages,
also form a strict subclass of \sCFHA languages.
Every \HA can indeed be presented as a \sCFHA
$\A = (\F, Q, Q^\final, \Delta)$ 
with variable-free transitions constrained with a type discipline:
$Q = \Q_\hor \uplus \Q_\ver$
and every transition of $\Delta $ has one of the forms
\[
\emptyhedge \to q_\hor \quad 
q_\hor . q_\ver \to q'_\hor \quad
a(q_\hor) \to q_\ver \quad
\mathrm{where}\; q_\hor, q'_\hor \in Q_\hor, q_\ver \in Q_\ver,  a \in \F.
\]

%
\noindent
From now on, we shall always consider \HA and \CFHA presented as \sCFHA.
%
%
%

\noindent
The following example shows that \sCFHA can capture some CF ranked tree languages.
Capturing the whole class of 
\begin{ABS}
CF RTL
\end{ABS}
\begin{RR}
CF ranked tree language 
\end{RR}
would require however a 
further generalization where permutations of variables are possible in 
the horizontal transitions of \sCFHA.
Such a generalization is out of the scope of this paper.
\begin{example} \label{ex:CF-ranked}
The language $\{ h^n(g(a^n(0), b^n(0))) \mid n \geq 1 \}$ is generated by 
the CF ranked tree grammar~\cite{TATA} with non-terminals $A$ and $S$ ($S$ is the axiom)
and productions 
\begin{RR}
rules 
\end{RR}
$A(x_1,x_2) \to h\bigl(A(a(x_1), b(x_2))\bigr)$,
$A(x_1,x_2) \to g(x_1,x_2)$ and $S \to A(0,0)$.
It is also recognized by the \sCFHA 
with transition rules 
$a(x_1).b(x_2) \to q(x_1 . x_2)$,
$g(x_1) \to q_0(x_1)$,
$q_0(q(x)) \to q_1(x)$, 
$h(q_1(x)) \to q_0(x)$
($q_0$ is final).
\end{example}


\subsection{Properties}
%
The class of \sCFHA language is
closed under union (direct construction by disjoint union of automata)
and not closed under intersection or complementation
(because CF word languages are defined by \sCFHA without vertical transitions).
%
\begin{property}
The membership problem is decidable for \sCFHA.
\end{property}
\begin{proof}
Let $h \in \H(\F)$ be a given hedge and $\A$ be a given \sCFHA.
We assume \textit{wlog} that $\A$ is presented as a set $\Delta$ of transitions in the above alternative form
$p_1(\delta_1) \ldots p_n(\delta_n) \to q(\delta_1\ldots \delta_n)$,  with $n > 0$,
and $p_1(p_2(\delta_1)) \to q(\delta_1)$.

Moreover, we assume that every transition of the form $q_1(x_1) \to q_2(x_1)$,
where $q_1$ and $q_2$ are states, has been removed,
replacing arbitrarily $q_1$ by $q_2$ in the \textit{rhs} of the other transitions.
Similarly, we remove $q_1 \to q_2$, replacing arbitrarily \textit{rhs}'s of the form $q_1$ by $q_2$.
All these transformations increase the size of $\A$ polynomialy.

Then all the horizontal transitions with $n=1$ have the form $a(\delta_1) \to q(\delta_1)$, with $a \in \F$.
It follows that the application of every rule of $\Delta$ 
strictly reduces the measure on hedges defined as pair
(\# of occurrences of symbols of $\F$, \# of occurrences of state symbols), ordered lexicographically.
During a reduction of $h$ by $\Delta$, each of the two components of the 
above measure is bounded by the size of $h$.
It follows that the membership $h \in L(\A)$ can be tested in PSPACE.
\qed
\end{proof}

\begin{property}
The emptiness problem is decidable in PTIME for \sCFHA.
\end{property}
\begin{proof}
Let $\A = \< \F, Q, Q^\final, \Delta>$.
We use a marking algorithm with two marks: $\hor$ and~$\ver$.
First, for technical convenience, 
we mark every symbol in $\F$ with $\ver$.
Then we iterate the following operations until no marking is possible
(note that the marking is not exclusive: some states may have 2 marks $\hor$ and $\ver$).

\noindent
For all transition 
$p_1(x_1) \ldots p_n(x_n) \to q(x_1\ldots x_n)$ in $\Delta$
such that every $p_i$ is marked,
if at least one $p_i$ is marked with $\ver$, then mark $q$ with $\ver$,
otherwise 
\begin{RR}
($n = 0$ or every $p_i$ is marked with $\hor$), 
\end{RR}
mark $q$ with~$\hor$.

\noindent
For all transition 
$p_1\bigl(p_2(x)\bigr) \to q(x)$ in $\Delta$
such that $p_1$ is marked $\ver$,
if $p_2$ is marked with $\ver$, then mark $q$ with $\ver$,
otherwise, 
if $p_2$ is marked with $\hor$, then mark $q$ with~$\hor$.

The number of iterations is at most $2.|Q|$ 
and the cost of each iteration is linear in the size of $\A$.
Then 
$q \in Q$ is marked with $\hor$ only iff
there exists $h \in \H(\F)$ such that $h \lrstep{*}{\Delta} q$,
and
it is marked with $\ver$ iff
there exists $C[\,] \in \C(\F)$ such that for all $h \in \H(\F)$, 
$C[h] \lrstep{*}{\Delta} q(h)$.
\begin{ABS}
Hence
\end{ABS}
\begin{RR}
It follows that 
\end{RR}
$L(\A) = \emptyset$ iff no state of $Q^\final$ is marked.
\qed
\end{proof}

For comparison, for  both classes of \HA and \CFHA,
the membership and emptiness problems are decidable 
in PTIME,
the class of \HA languages is closed under Boolean operations and
the class of \CFHA languages is closed under union but not closed under intersection and complementation,
see \cite{Murata00,OhsakiST03,TATA}.
%

%

%
%


\section{Inverse Monadic Hedge Rewriting Systems} \label{sec:TRS}
A rewrite rule $\ell \to r$ over $\F$ is called \emph{monadic} 
(following~\cite{Salomaa88,CoquideDauchetGilleronVagvolgyi94})
if $r = a(x)$ with $a \in \F$, $x \in \X$,
\emph{inverse-monadic} 
if $r \to \ell$ is monadic and $r \notin \X \cup \{ \emptyhedge \}$,
and \emph{1-childvar} if it contains at most one variable and this variable
has no siblings in $\ell$ and $r$.
%
%
Intuitively, every finite, linear, inverse-monadic, 1-childvar \TRS
can be transformed into a \TRS equivalent \textit{wrt} reachability
whose rules are inverse of transitions of \sCFHA.
It follows that such \TRS preserve \sCFHA languages.

\begin{example}
The \TRS of Example~\ref{ex:TRS} is linear, inverse-monadic, and 1-childvar.
The closure of the language $\{ p_0 \}$ is the \sCFHA language
of T-patterns.
\end{example}

\begin{theorem} \label{th:post*-monadic}
Let $L$ be the language of $\A_L \in$ \sCFHA, 
and $\R$ be a finite, linear, inverse-monadic, 1-childvar \TRS.
There exists an effectively computable \sCFHA recognizing $\post^*_{\R}(L)$,
of size polynomial in the size of $\R$ and $\A_L$.
\end{theorem}
\begin{proof}
Let $\A_L = \< \F, Q_L, Q_L^\final, \Delta_L>$, 
we construct a \sCFHA $\A = \< \F, Q, Q^\final, \Delta>$.
The state set $Q$ contains all the states of $Q_L$,
one state $\underline{h}$ for every non-variable sub-hedge of a \textit{rhs} of rule of $\R$,
one state $\underline{a}$ for each $a \in \F$ and one new state $q \notin Q_L$.
For each $p \in Q_L \cup \F$, we note $\underline{p} = \underline{a}$
if $p = a \in \F$ and $\underline{p} = p$ otherwise.
Let $Q^\final = Q_L^\final$ and let $\Delta_0$ contain the following transition rules,
where $a \in \F$, $t \in \T(\F, \{ x \})$ and $h \in \H(\F, \{ x \}) \setminus \{ \emptyhedge \}$. 
\[
\begin{array}{c}
\begin{array}{rcll}
\underline{p_1}(x_1) \ldots \underline{p_n}(x_n) & \to & q(x_1\ldots x_n) & 
\mathrm{if}\; p_1(x_1) \ldots p_n(x_n) \to q(x_1\ldots x_n) \in \Delta_L\\
\underline{p_1}\bigl(\underline{p_2}(x)\bigr) & \to & q(x) & \mathrm{if}\; p_1\bigl(p_2(x)\bigr) \to q(x) \in \Delta_L\\
\end{array}
\\[3mm]
\begin{array}{ccc}
\begin{array}{rcll}
\underline{t}(x) . \underline{h} & \to & \underline{t . h}(x)  & \mathrm{if}\; x \in \var(t), 
                                                                   \underline{t . h} \in Q\\
\underline{t}(x) . \underline{h} & \to & q(x)  & \mathrm{if}\; x \in \var(t), 
                                                                   \underline{t . h} \notin Q\\
\underline{t} . \underline{h}(x) & \to & \underline{t . h}(x) & \mathrm{if}\; x \notin \var(t),
                                                                  \underline{t . h} \in Q\\
\underline{t} . \underline{h}(x) & \to & q(x) & \mathrm{if}\; x \notin \var(t),
                                                                  \underline{t . h} \notin Q\\
\end{array}
& &
\begin{array}{rcll}
a(x)                  & \to & \underline{a}(x)                &  \\
a(\underline{h}(x))   & \to & \underline{a(h)}(x)             &  \mathrm{if}\; \underline{a(h)} \in Q\\
a(\underline{h}(x))   & \to & \underline{a}(x)                &  \mathrm{if}\; \underline{a(h)} \notin Q\\
a(q(x))               & \to & \underline{a}(x)                &  \\ 
\end{array}
\end{array}
\end{array}
\]
Finally let $\Delta = \Delta_0 \cup \{ \underline{h}(x) \to \underline{a}(x) \mid a(x) \to h \in \R \}$.
Let $\ell \in \H(\F)$ be such that $\ell \lrstep{*}{\Delta} s(u)\ (\star)$, with $s \in Q$ and $u\in \H(Q \cup \F)$.
We show by induction on the number $N$ of applications of rules of $\Delta \setminus \Delta_0$ in $(\star)$
that there exists $\ell' \in \H(\F)$ such that $\ell' \lrstep{*}{\R} \ell$ and moreover, 
if $s = \underline{h}$, then $h$ matches $\ell'$, 
if $s = q$ then $\ell'$ is not matched by a non-variable subhedge of \textit{rhs} of rule of $\R$ and
if $s \in Q_L$, then $\ell' \in L(\A_L, s)$.

If $N = 0$, then the property holds with $\ell' = \ell$
(this can be shown by induction on the length of $(\star)$).
If $N > 0$, we can assume that $(\star)$ has the following form.
\[
\ell = C[k] \lrstep{*}{\Delta_0} C[\underline{h}(v)] \lrstep{}{\Delta \setminus \Delta_0}
       C[\underline{a}(v)] \lrstep{}{\Delta} s(u)
\]
It follows that $h$ matches $k$, i.e. there exists $w$ such that $k = h[w]$,
and $w \lrstep{*}{\Delta_0} v$.
Hence $\ell' = C[a(w)] \lrstep{}{\R} \ell$, 
and $\ell' \lrstep{*}{\Delta_0} C[a(v)] \lrstep{}{\Delta_0} C[\underline{a}(v)] \lrstep{}{\Delta} s(u)$.
We can then apply the induction hypothesis to $\ell'$, and 
immediately conclude for~$\ell$.
\qed
\end{proof}

The following Example~\ref{ex:gen} illustrates the importance of the 1-childvar 
and condition in Theorem~\ref{th:post*-monadic}.
\begin{example} \label{ex:gen}
With the following rewrite rule 
$a(x) \to c\,a(e\,x\,g)\,d$
we generate from $\{ a \}$ the language
$\{ c^n a(e^n g^n)\, d^n \mid n \geq 1 \}$,
seemingly not \sCFHA.
\end{example}

\medskip
In \cite{JR-rta2008} it is shown that the closure of a \HA language
under rewriting with a monadic \TRS 
is a \HA language.
It follows that the backward rewrite closure of a \HA language under an inverse-monadic \TRS is \HA.


\section{Update Hedge Rewriting Systems} \label{sec:PTRS}
In this section, we turn to our 
motivation of studying XQuery Update Facility primitives
modeled as parameterized rewriting rules.

Let $\A = \<\F, Q, Q^\final, \Delta>$ be a \HA.
A hedge rewriting system over $\F$ parame\-tri\-zed by $\A$ 
(\textbf{PHRS})
is given by a finite set, denoted $\ptrs{\R}{\A}$, of 
rewrite rules $\ell \to r$ 
where 
$\ell \in \H(\F,\X)$ and 
$r \in \H(\F \uplus Q,\X)$ and symbols of $Q$ can only label leaves of $r$
($\uplus$ stands disjoint union, hence we implicitly 
assume that $\F$ and $Q$ are disjoint sets).
In this notation, $\A$ may be omitted when it is clear from context
or not necessary.
The rewrite relation $\lrstep{}{\ptrs{\R}{\A}}$ associated to a \PTRS $\ptrs{\R}{\A}$ is
defined as the rewrite relation $\lrstep{}{\R[\A]}$ 
where the \TRS $\R[\A]$ is the (possibly infinite) set of all rewrite rules
obtained from rules $\ell \to r$ in $\ptrs{\R}{\A}$
by replacing in $r$ every state $p \in Q$ by a ground hedge of $L(\A, p)$. 
Note that
when there are multiple occurrences of a state $p$ in a rule,
each occurrence of $p$ is independently replaced with
a hedge in $L(\A, p)$, which can generally be different from one another.
Given a set $L \subseteq \H(\F,\X)$, 
we define $\post_{\ptrs{\R}{\A}}^*(L)$ to be $\post_{\R[\A]}^*(L)$%
\begin{ABS}
.
\end{ABS}
\begin{RR}
and $\pre_{\ptrs{\R}{\A}}^*(L)$ symmetrically. 
\end{RR}

\noindent
We call \emph{updates} parametrized rewrite rules of the following form
%
\[
\begin{array}{rcllclc}
a(x) & \to & b(x)               & & \quad & \mathrm{node~renaming} & (\REN)\\
a(x) & \to & a(u_1 \, x \, u_2) &\ u_1, u_2 \in Q^* & & \mathrm{addition~of~child~nodes} & (\CHILD)\\
a(x) & \to & v_1 \, a(x) \, v_2 &\ v_1, v_2 \in Q^* & & \mathrm{addition~of~sibling~nodes} & (\SIB)\\
a(x) & \to & b\bigl(a(x)\bigr)  & & & \mathrm{addition~of~parent~node} & (\PAR)\\
a(x) & \to & u                  &\ u \in Q^* & & \mathrm{node~replacement/recursive~deletion}\quad & (\RPL)\\
a(x) & \to & x                  & & & \mathrm{single~node~deletion} & (\DEL)\\
\end{array}
\]
Note that the particular case of ($\RPL$)  of $\RPL$ with $u = \emptyhedge$
corresponds to the deletion of the whole subtree $a(x)$.
In the rest of the paper, a \PTRS containing only updates 
will be called update \PTRS (\textbf{uPHRS}).


\subsection{Loop-free \uPTRS}
In order to simplify the proofs 
we can reduce 
to the case where there exists no looping sequence of renaming. 
This motivates the following definition: 
\begin{definition}\label{def:loopfree}
An \uPTRS $\ptrs{\R}{\A}$ is {\em loopfree} if there exists no sequence 
$a_1,\ldots, a_n$ ($n >1$) such that  for all $1\leq i<n$, 
\begin{ABS}
$a_i(x)\to a_{i+1}(x) \in \R$ 
\end{ABS}
\begin{RR}
$a_i(x)\to a_{i+1}(x)$ is a renaming rule of $\R$ 
\end{RR}
and $a_1 = a_n$. 
\end{definition}
Given a \uPTRS $\ptrs{\R}{\A}$, 
we consider the directed graph $G$ whose set of nodes is $\F$ 
and containing an edge $\< a,b>$ iff $a(x)\to b(x)$ is in $\R$. 
For every strongly connected component in $G$ we select a representative. 
We denote by $\hat{a}$ the representative of $a$ in its component and more generally 
by $\hat{h}$ the hedge obtained from $h \in \H(\F)$ 
by replacing every function symbol $a$ by its representative $\hat{a}$. 
We define $\hat{\R}$ to be $\R$ where every rule $\ell \to r$ is replaced 
by $\hat{\ell}\to \hat{r}$ (if the two members get equal we can remove the rule). 
We define $\hat{\A}$ analogously. 

\begin{lemma}
Given an  \uPTRS $\ptrs{\R}{\A}$  the  \uPTRS $\ptrs{\hat{R}}{\hat{A}}$ is loopfree and  
for all $h,h'\in\H(\F)$ we have $h \lrstep{*}{\ptrs{\R}{\A}}h'$ iff 
$\hat{h}\lrstep{*}{\ptrs{\hat{R}}{\hat{A}}}\hat{h}'$.
\end{lemma}
\begin{proof}
By induction on the length of derivations.\qed
\end{proof}

\subsection{Rewrite Closure}
The rest of the section is devoted to the proof of the following theorem
of construction of \sCFHA for the forward closure by updates.
\begin{theorem} \label{th:post*-update}
Let $\A$ be a \HA over $\F$, 
and $L$ be the language of $\A_L \in$ \CFHA, 
and $\ptrs{\R}{\A}$ be a loop-free \uPTRS.
There exists an effectively computable \CFHA recognizing $\post^*_{\ptrs{\R}{\A}}(L)$,
of size polynomial in the size of $\ptrs{\R}{\A}$ and $\A_L$ and
exponential in the size of the alphabet $\F$.
\end{theorem}

\noindent
%
The construction of the \CFHA works in 2 steps: 
construction of an initial automaton and completion loop.
\begin{ABS}
We shall use the following notion in order to simplify the proof:
\end{ABS}
\begin{RR}
We need first a notion of normalization of \CFHA in order to simplify the proofs:
\end{RR}
a \CFHA $\<\F, Q, Q^\final, \Delta>$  is called
\emph{normalized} if for all $a \in \F$ and $q \in Q$, 
there exists one unique state of $Q$ denoted $\initial{q}{a}$
such that $a(\initial{q}{a}) \to q \in \Delta$,
and moreover, $\initial{q}{a}$ does neither occur in a
left hand side of an horizontal transition of~$\Delta$
nor in a right hand side of a vertical transition of $\Delta$.
\begin{ABS}
With some state renaming, every \CFHA $\A$ can be transformed in PTIME into a normalized \CFHA $\A'$, 
of size linear in the size of $\A$, and such that $L(\A') = L(\A)$.
\end{ABS}
\begin{RR}
\begin{lemma}[Normalization]
For all \CFHA $\A$, there exists a normalized \CFHA $\A'$ 
such that $L(\A') = L(\A)$, 
of size linear in the size of $\A$ and which can be constructed in PTIME.
\end{lemma}
%
\end{RR}

\paragraph{Initial automaton.}
Let $\A = \< \F, Q_\A, Q_\A^\final, \Delta_\A>$
and $\A_L = \< \F, Q_L, Q_L^\final, \Delta_L>$.
We assume that the state sets $Q_\A$ and $Q_L$ are disjoint.
\begin{RR}
We will construct a \sCFHA $\A'$ 
for the recognition of $\post_{\ptrs{\R}{\A}}^*(L)$.
\end{RR}

First, 
\begin{RR}
in order to simplify the construction, 
\end{RR}
let us merge $\A$ and $\A_L$ into 
a \CFHA $\B = \<\F, P, P^\final, \Gamma>$
obtained by the normalization of 
$\<\F, Q_\A \uplus Q_L, Q_L^\final, \Delta_\A \uplus \Delta_L>$.
Below, the states of $P$ will be 
denoted by the letters $p$ or $q$.
Let $\Pin$ be the subset of states of $P$ 
of the form $\initial{q}{a}$
(remember that $\initial{q}{a}$ is a state of $P$ 
 uniquely characterized by $a \in \F$, $q \in P$, since $\B$ is normalized).
We assume \textit{wlog} that $\Pin$ and $P^\final$ are disjoint
and that $\B$ is \emph{clean}, \textit{i.e.} 
for all $p \in P$, $L(\B, p) \neq \emptyset$.

Next, in a preliminary construction step, we transform the initial automaton 
$\B$ into a \CFHA $\A_0 = \<\F, Q, Q^\final, \Delta_{0}>$.
Let us call \emph{renaming chain} a sequence $a_1,\ldots,a_n$ of symbols of $\F$ 
such that $n \geq 1$ for all $1 \leq i < n$, $a_i(x) \to a_{i+1}(x) \in \R$.
Since $\R$ is loop-free, the length of every renaming chains is bounded by $|\F|$.
The fresh state symbols of $Q$ are 
defined as extensions of the symbols of $P \setminus \Pin$
with renaming chains.
We consider two modes for such states:
the \emph{push} and \emph{pop} modes, characterized by 
a chain respectively in superscript or subscript.
\[ 
Q = P \cup \{ \ustack{a}{q} \mid \stack{a}{q} \in \Pin \} 
      \cup 
 \left\{ 
 \begin{array}{l}
 \stack{a_1\ldots a_n}{q}  \bigm| q \in P \setminus \Pin, 
                                  n \geq 2,\\
 \ustack{a_1\ldots a_n}{q} \bigm| 
 a_1,\ldots, a_n \; \mathrm{is~a~renaming~chain}\\
 \end{array}
\right\}
\]
Let 
$Q^\final = P^\final$ be the subset of final states.
Intuitively, in the state $\stack{a_1\ldots a_n}{q}$, 
the chain of $\Sigma^+$ represents a sequence 
of renamings, with $\ptrs{\R}{\A}$, of the parent of the current symbol,
starting with $a_1$ and ending with $a_n$.
Note that the states of $\Pin$ are particular cases of such states, 
with a chain of length one.
A state $\ustack{a_1\ldots a_n}{q}$
will be used below to represent the tree $a_n(\stack{a_1\ldots a_n}{q})$.

\noindent
The initial set of transitions $\Delta_0$ is defined as follows
\[
\begin{array}{rcl}
\Delta_0 = \Gamma_h  & \cup & \{ \ustack{a_1}{q} \to q \mid \ustack{a_1}{q} \in Q \}\\
& \cup &
\{  
 a_n\bigl( \stack{a_1\ldots a_n}{q} \bigr) \to \ustack{a_1\ldots a_n}{q} 
\mid \stack{a_1\ldots a_n}{q}, \ustack{a_1\ldots a_n}{q} \in Q, n \geq 1
\}
\end{array}
\]
where $\Gamma_h$ is the subset of horizontal transitions of $\Gamma$.
Note that $\A_0$ is not normalized.
\noindent
The following lemma is immediate by construction of $\Gamma$ and $\A_0$.
%

\begin{lemma} \label{cor:A0}
For all $q \in Q_\A$ (resp.  $q \in Q_L$)
$L(\A_0, q) = L(\A, q)$ 
(resp. $L(\A_L, q)$).
\end{lemma}
\begin{proof}
Every vertical transition in $\Gamma$ has the form
$a(\stack{a}{q}) \to q$ and can be simulated by the 2 steps
$a(\stack{a}{q}) \to \ustack{a}{q} \to q$.
Moreover, all the states $\stack{a_1\ldots a_n}{q}$ and $\ustack{a_1\ldots a_n}{q}$
with $n \geq 2$ are empty for $\A_0$.
\qed
\end{proof}

\noindent
For the construction of $\A'$,
we shall complete incrementally $\Delta_0$ into $\Delta_1$, $\Delta_2$,...
by adding some transition rules, according to a case analysis 
of the rules of $\ptrs{\R}{\A}$.
For each construction step $i \geq 0$, we let
\( \A_i = \<\F, Q, Q^\final, \Delta_i> \).

\paragraph{Automata completion.}
The construction of the sequence $(\Delta_i)$
works by iteration of a case analysis of the rewrite rules of $ \ptrs{\R}{\A}$, 
presented in Table~\ref{tab:completion}.
Assuming that $\Delta_i$ is the last set built, 
we define its extension $\Delta_{i+1}$ 
by application of the first case in Table~\ref{tab:completion} such that 
$\Delta_{i+1} \neq \Delta_i$. 
In the rules of Table~\ref{tab:completion}, 
$a_1,\ldots,a_n, b$ are symbols of $\F$, and $u, v$ are sequences of $Q_\A^*$.

\begin{table}
\[
\begin{array}{|c|l|l|}
\hline
 & \ptrs{\R}{\A}\;\mathrm{contains} & 
\multicolumn{1}{c|}{\Delta_{i+1} = \Delta_i \cup}\\
\hline
(\REN) & \; a_n(x) \to b(x) & 
\begin{array}{cl}
      & \{ \stack{a_1\ldots a_n}{q} \to \stack{a_1\ldots a_n b}{q} 
           \mid \stack{a_1\ldots a_n b}{q} \in Q  \}\\
 \cup & \{ \ustack{a_1\ldots a_n b}{q} \to \ustack{a_1\ldots a_n}{q} 
           \mid \ustack{a_1\ldots a_n b}{q} \in Q  \}
\end{array}  
\\[2mm]
(\CHILD) & \; a_n(x) \to a_n(u\, x \, v) & 
\quad \{ u \, \stack{a_1\ldots a_n}{q} \, v \to \stack{a_1\ldots a_n}{q}
         \mid \stack{a_1\ldots a_n}{q} \in Q \} 
\\[1mm]
(\SIB) & \; a_n(x) \to u\, a_n(x) \, v & 
\quad \{ u \, \ustack{a_1\ldots a_n}{q} \, v \to \ustack{a_1\ldots a_n}{q}
         \mid \ustack{a_1\ldots a_n}{q} \in Q \} 
\\[1mm]
(\PAR) & \; a_n(x) \to b\bigl(a_n(x)\bigr) & 
\quad \{ b\bigl(\ustack{a_1\ldots a_n}{q}\bigr) \to \ustack{a_1\ldots a_n}{q}
         \mid \ustack{a_1\ldots a_n}{q} \in Q \} 
\\[1mm]
(\RPL) & \; a_n(x) \to u & 
\quad \{ u \to \ustack{a_1\ldots a_n}{q} \mid \ustack{a_1\ldots a_n}{q} \in Q \} 
\\[1mm]
(\DEL) & \; a_n(x) \to x & 
\quad \{ \stack{a_1\ldots a_n}{q} \to \ustack{a_1\ldots a_n}{q}
         \mid \stack{a_1\ldots a_n}{q} \in Q \} 
\\[1mm]
\hline
\end{array}
\]
\caption{\CFHA Completion}
\label{tab:completion}
\end{table}

%
Only a bounded number of rules can be added to the $\Delta_i$'s, hence eventually, 
a fixpoint $\Delta_k$ is reached, that we will denote $\Delta'$.
We also write $\A'$ for $\A_k$.


\begin{RR}
\subsubsection{Correctness.}
\end{RR}
\begin{ABS}
\medskip
\end{ABS}

The following Lemma~\ref{lem:correctness}
shows that the automata computations simulate
the rewrite steps, 
\textit{i.e.}  that $L(\A') \subseteq \post^*_{\ptrs{\R}{\A}}(L)$.
Let us abbreviate $\ptrs{\R}{\A}$ by $\R$.
We use the notation 
$h \lrstep{a_1\ldots a_n}{\R} h'$, 
for a renaming chain $a_1,\ldots,a_n$ ($n \geq 1$), 
%
%
if there exists $h_1, \ldots h_{n} \in \H(\F)$ such that
\[
h = a_1(h_1) 
\lrstep{*}{\R}  a_1(h_2) 
\lrstep{}{\REN} a_{2}(h_2)
\lrstep{*}{\R} \ldots 
\lrstep{*}{\R}  a_{n-1}(h_n)
\lrstep{}{\REN} a_n(h_n) \lrstep{*}{\R} h'
\]
where the reductions denoted $\lrstep{}{\REN}$ are 
rewrite steps with rules of $\ptrs{\R}{\A}$ of type $(\REN)$,
applied at the positions of $a_1$,\ldots, $a_n$,
and all the other rewrite steps (denoted $\lrstep{*}{\R}$)
involve no rule of type $(\REN)$.
%

\begin{lemma}[Correctness] \label{lem:correctness}
For all $h \in \H(\F)$,
\begin{description}
\item[$i.$]
if $h \lrstep{*}{\A'} \ustack{a_1\ldots a_n}{q}$, with $n \geq 1$, 
then there exists $h_1 \in \H(\F)$ such that\\
$a_1(h_1) \lrstep{*}{\B} q$ and $a_1(h_1) \lrstep{a_1\ldots a_n}{\R} h$,
\item[$ii.$]
if $h \lrstep{*}{\A'} \stack{a_1\ldots a_n}{q}$, with $n \geq 1$, 
then there exists $h_1 \in \H(\F)$ such that\\
$h_1 \lrstep{*}{\B} \stack{a_1}{q}$, 
and 
$a_1(h_1) \lrstep{a_1\ldots a_n}{\R} a_n(h)$,
%
\item[$iii.$]
if $h \lrstep{*}{\A'} q\in P\setminus \Pin$, 
then there exists $h' \in \H(\F)$ such that\\
$h' \lrstep{*}{\B} q$ and $h' \lrstep{*}{\R} h$.
\end{description}
\end{lemma}
\begin{proof}
\begin{ABS}
(sketch)
\end{ABS}
Let $s \in Q$ be such that $h \lrstep{*}{\A'} s$ and let us call $\rho$ this reduction.
With a commutation of transitions, we can assume that $\rho$ has the following form,
\[
\rho: h  = t_1\ldots t_m \lrstep{*}{\A'} \underbrace{s_1\ldots s_m \lrstep{*}{\A'} s}_{\textstyle\rho_0}
\]
where $t_1,\ldots,t_m \in \T(\F)$, $s_1,\ldots,s_m \in Q$,
and for all $1 \leq i \leq m$,
$t_i \lrstep{*}{\A'} s_i$, and the last step of this reduction
involves a vertical transition
$a(\stack{a_1\ldots a_n}{q}) \to s_i$ or
$b(\ustack{a_1\ldots a_n}{q}) \to s_i$.
\noindent
The proof is by induction on the length of $\rho$.

\medskip\noindent
The shortest possible $\rho$
has 2 steps: $h = t_1 = a(\emptyhedge) \lrstep{}{\A_0} a(\stack{a}{q}) \lrstep{}{\A_0} q = s$
and ($iii$) holds immediately with $h' = h$, by Lemma~\ref{cor:A0}.

\medskip\noindent
For the induction step, we consider the length of $\rho_0$.
\begin{RR}

\end{RR}
\noindent
If $|\rho_0| = 0$, 
we have necessarily $m=1$, and the reduction $\rho$ has one of the two following forms
($\vec{v} \in Q^*$).
\begin{eqnarray}
h & = & t_1 = b(h') \lrstep{*}{\A'} b(\vec{v}) \lrstep{*}{\A'} 
b(\ustack{a_1\ldots a_n}{q}) \lrstep{}{\A'} \ustack{a_1\ldots a_n}{q} = s_1 = s
\label{eq:completion-par}\\
h & = & t_1 = a_n(h') \lrstep{*}{\A'} a_n(\vec{v}) \lrstep{*}{\A'} 
a_n(\stack{a_1\ldots a_n}{q}) \lrstep{}{\A_0} \ustack{a_1\ldots a_n}{q} = s_1 = s
\label{eq:completion-0}
\end{eqnarray}

In the case (\ref{eq:completion-par}), 
assume that the vertical transition
$b(\ustack{a_1\ldots a_n}{q}) \to \ustack{a_1\ldots a_n}{q}$ 
has been added to $\A'$ because $\ptrs{\R}{\A}$ contains a rule 
$a_n(x) \to b\bigl(a_n(x)\bigr)$.
By induction hypothesis ($i$) applied to the sub-reduction
$h' \lrstep{*}{\A'} \ustack{a_1\ldots a_n}{q}$,
there exists $h_1 \in \H(\F)$ such that
$a_1(h_1) \lrstep{*}{\B} q$,
and $a_1(h_1) \lrstep{a_1\ldots a_n}{\R} h'$.
It follows in particular that there exists $h_n$ such that
$a_n(h_n) \lrstep{*}{\R} h'$, 
and using the above $(\PAR)$ rewrite rule, 
$a_n(h_n) \lrstep{}{\R} b\bigl(a_n(h_n)\bigr) \lrstep{*}{\R} b(h') = h$.
Therefore, $a_1(h_1) \lrstep{a_1\ldots a_n}{\R} h$
and ($i$) holds for $h$ and $s$.

In the case (\ref{eq:completion-0}), 
by induction hypothesis ($ii$) applied to the sub-reduction
$h' \lrstep{*}{\A'} \stack{a_1\ldots a_n}{q}$,
there exists $h_1 \in \H(\F)$ such that
$h_1 \lrstep{*}{\B} \stack{a_1}{q}$,
hence $a_1(h_1) \lrstep{*}{\B} q$,
and $a_1(h_1) \lrstep{a_1\ldots a_n}{\R} a_n(h') = h$.
Therefore ($i$) holds for $h$ and $s$.

\medskip\noindent
Assume now that $|\rho_0| > 0$, 
and let us analyze the horizontal transition rule used in 
the last step of $\rho_0$.
\begin{ABS}
In order to comply with spaces restrictions, we will present only one significant case in this extended abstract (see~\cite{versionlongue} for the other cases).
\end{ABS}

\begin{RR}
\paragraph{Case $\Delta_0.1$} 
The last step of $\rho_0$ involves 
$q_1 \ldots q_n \to q \in \Gamma_h$ 
(horizontal transition of $\B$), with $n \geq 0$.
In this case, the reduction $\rho$ has the form
\[
h = h_1\ldots h_n \lrstep{*}{\A'} s_1\ldots s_m \lrstep{*}{\A'} 
q_1 \ldots q_n \lrstep{}{\A_0} q = s
\]
with $n \leq m$, $h_i \in \H(\F)$ and 
$h_i \lrstep{*}{\A'} q_i$ for all $i \leq n$.
By induction hypothesis ($iii$) applied to the latter reductions, 
for all $i \leq n$, there exists $h'_i$ such that 
$h'_i \lrstep{*}{\B} q_i$ and $h'_i \lrstep{*}{\R} h_i$.
Hence ($iii$) holds for $h$ and $s$ with $h' = h'_1\ldots h'_n$, 
since 
$h' \lrstep{*}{\B} q_1\ldots q_n \lrstep{}{\B} q$, and 
$h' \lrstep{*}{\R} h$.

\paragraph{Case $\Delta_0.2$} 
The last step of $\rho_0$ uses 
$\ustack{a_1}{q} \to q \in \Delta_0$.
In this case, the reduction $\rho$ has the form
\[
h \lrstep{*}{\A'} \ustack{a_1}{q}  \lrstep{}{\A_0} q = s
\]
\noindent
By induction hypothesis ($i$) applied to $h \lrstep{*}{\A'} \ustack{a_1}{q}$,
there exists $h_1 \in \H(\F)$ such that 
$a_1(h_1) \lrstep{*}{\B} q$ and
$a_1(h_1) \lrstep{a_1}{\R} h$.
Hence, ($iii$) holds with $h' = a_1(h_1)$.

\paragraph{Case $(\REN).1$} 
The last step of $\rho_0$ uses 
$\stack{a_1\ldots a_{n-1}}{q} \to \stack{a_1\ldots a_{n}}{q}$ and 
this transition has been added to $\Delta'$
because $\ptrs{\R}{\A}$ contains a rule $a_{n-1}(x) \to a_n(x)$.
In this case, the reduction $\rho$ has the form
\begin{equation}
h \lrstep{*}{\A'} \stack{a_1\ldots a_{n-1}}{q} \lrstep{}{\A'} \stack{a_1\ldots a_{n}}{q} = s
\label{eq:correctness-REN}
\end{equation}

By induction hypothesis ($ii$) applied to $h \lrstep{*}{\A'} \stack{a_1\ldots a_{n-1}}{q}$,
there exists $h_1 \in \H(\F)$ such that
$h_1 \lrstep{*}{\B} \stack{a_1}{q}$, 
and $a_1(h_1) \lrstep{a_1\ldots a_{n-1}}{\R} a_{n-1}(h)$.
Since by hypothesis $a_{n-1}(h) \lrstep{}{\R} a_n(h)$, 
we have $a_1(h_1) \lrstep{a_1\ldots a_n}{\R} a_{n}(h)$ and ($ii$) holds for $h$ and $s$.


\paragraph{Case $(\REN).2$} 
The last step of $\rho_0$ uses 
$\ustack{a_1\ldots a_{n}}{q} \to \ustack{a_1\ldots a_{n-1}}{q}$ 
and this transition has been added to $\Delta'$
because $\ptrs{\R}{\A}$ contains a rule $a_{n-1}(x) \to a_n(x)$.
In this case, the reduction $\rho$ has the form
\begin{equation}
h \lrstep{*}{\A'} \ustack{a_1\ldots a_{n}}{q} \lrstep{}{\A'} \ustack{a_1\ldots a_{n-1}}{q} = s
\label{eq:correctness-REN}
\end{equation}

By induction hypothesis ($i$) applied to $h \lrstep{*}{\A'} \ustack{a_1\ldots a_{n}}{q}$,
there exists $h_1 \in \H(\F)$ such that
$a_1(h_1) \lrstep{*}{\B} q$, 
and $a_1(h_1) \lrstep{a_1\ldots a_n}{\R} h$.
It follows immediately by definition that
$a_1(h_1) \lrstep{a_1\ldots a_{n-1}}{\R} h$.

\end{RR}

\paragraph{Case $(\CHILD)$.}
The last step of $\rho_0$ uses 
$u \, \stack{a_1\ldots a_n}{q} \, v \to \stack{a_1\ldots a_n}{q}$
and this transition has been added to $\Delta'$
because $\ptrs{\R}{\A}$ contains a rule $a_n(x) \to a_n(u\, x \, v)$, 
with $u, v \in Q_\A^*$.
In this case, the reduction $\rho$ has the following form,
\begin{equation}
h = \ell\, h'\, r \lrstep{*}{\A'} u\,\stack{a_1\ldots a_n}{q}\, v \lrstep{}{\A'} \stack{a_1\ldots a_n}{q} = s
\label{eq:correctness-INSfirst}
\end{equation}
where $\ell \lrstep{*}{\A'} u$,
$h' \lrstep{*}{\A'} \stack{a_1\ldots a_n}{q}$, 
and $r \lrstep{*}{\A'} v$.
By induction hypothesis ($ii$) applied to $h' \lrstep{*}{\A'} \stack{a_1\ldots a_n}{q}$,
there exists $h_1$ such that 
$h_1 \lrstep{*}{\B} \stack{a_1}{q}$
and 
$a_1(h_1) \lrstep{a_1\ldots a_n}{\R} a_n(h')$,
and by induction hypothesis ($iii$) applied to $\ell \lrstep{*}{\A'} u$ (resp. $r \lrstep{*}{\A'} v$),
and by Lemma~\ref{cor:A0},
there exists $\ell' \in \H(\F)$ (resp. $r' \in \H(\F)$) 
such that $\ell' \lrstep{*}{\A} u$ (resp. $r' \lrstep{*}{\A} v$)
and $\ell' \lrstep{*}{\R} \ell$ (resp. $r' \lrstep{*}{\R} r$).
It follows that $a_n(h') \lrstep{}{\R} a_n(\ell'\, h'\, r') \lrstep{*}{\R} a_n(\ell\, h'\,r) = a_n(h)$.
Hence 
$a_1(h_1) \lrstep{a_1\ldots a_n}{\R} a_n(h)$ 
and ($ii$) holds for $h$ and $s$.
\begin{ABS}
\qed\end{proof}
\end{ABS}

\begin{RR}

\paragraph{Case $(\SIB)$.}
The last step of $\rho_0$ uses 
$u \, \ustack{a_1\ldots a_n}{q} \, v \to \ustack{a_1\ldots a_n}{q}$
and this transition has been added to $\Delta'$
because $\ptrs{\R}{\A}$ contains a rule $a_n(x) \to u\, a_n(x)\, v$, 
with $u, v \in Q_\A^*$.
In this case, the reduction $\rho$ has the following form,
\[
h = \ell\, h'\, r \lrstep{*}{\A'} u\,\ustack{a_1\ldots a_n}{q}\, v \lrstep{}{\A'} \ustack{a_1\ldots a_n}{q} = s
\]
where $\ell \lrstep{*}{\A'} u$,
$h' \lrstep{*}{\A'} \ustack{a_1\ldots a_n}{q}$, 
and $r \lrstep{*}{\A'} v$.
By induction hypothesis ($i$) applied to $h' \lrstep{*}{\A'} \ustack{a_1\ldots a_n}{q}$,
there exists $h_1$ such that 
$a_1(h_1) \lrstep{*}{\B} q$ and  $a_1(h_1) \lrstep{a_1\ldots a_n}{\R} h'$.
To be more precise, the latter reduction has the form
\[
a_1(h_1) 
\lrstep{*}{\R} a_1(h_2) 
\lrstep{}{\REN} a_2(h_2)
\lrstep{*}{\R} \ldots 
\lrstep{*}{\R} a_{n-1}(h_n)
\lrstep{}{\REN} a_n(h_n) \lrstep{*}{\R} h'
\]
for some $h_2,\ldots, h_n \in \H(\F)$.

Moreover, by induction hypothesis ($iii$) applied to $\ell \lrstep{*}{\A'} u$ 
(resp. $r \lrstep{*}{\A'} v$),
and by Lemma~\ref{cor:A0},
there exists $\ell' \in \H(\F)$ (resp. $r' \in \H(\F)$) 
such that $\ell' \lrstep{*}{\A} u$ (resp. $r' \lrstep{*}{\A} v$)
and $\ell' \lrstep{*}{\R} \ell$ (resp. $r' \lrstep{*}{\R} r$).
Therefore, 
$a_n(h_n) \lrstep{}{\R} \ell'\, a_n(h_n)\, r' \lrstep{*}{\R} \ell\, a_n(h_n)\,r \lrstep{*}{\R} \ell\, h'\,r = h$.
Hence $a_1(h_1) \lrstep{a_1\ldots a_n}{\R} h$ and ($i$) holds for $h$ and $s$.

\paragraph{Case $(\RPL)$.}
The last step of $\rho_0$ uses 
$u \to \ustack{a_1\ldots a_n}{q}$, 
and this transition has been added to $\Delta'$
because $\ptrs{\R}{\A}$ contains a rule $a_n(x) \to u$, with $u \in Q_\A^*$.
In this case, the reduction $\rho$ has the following form,
\[
h \lrstep{*}{\A'} u \lrstep{}{\A'} \ustack{a_1\ldots a_n}{q} = s
\]
By induction hypothesis ($iii$) applied to $h \lrstep{*}{\A'} u$ 
and by Lemma~\ref{cor:A0},
there exists $h' \in \H(\F)$
such that $h' \lrstep{*}{\A} u$ 
and $h' \lrstep{*}{\R} h$.
Since $\B$ is assumed clean, there exists $h_1 \in L(\B, \stack{a_1}{q})$,
and, using the above $(\RPL)$ rewrite rule, $a_n(h_1) \lrstep{}{\R} h' \lrstep{*}{\R} h$.
Hence 
$a_1(h_1) \lrstep{\*}{\B} q$ and
$a_1(h_1) \lrstep{a_1\ldots a_n}{\R} h$ and ($i$) holds for $h$ and $s$.

\paragraph{Case $(\DEL)$.}
The last step of $\rho_0$ uses 
$\stack{a_1\ldots a_n}{q} \to \ustack{a_1\ldots a_n}{q}$
and this transition has been added to $\Delta'$ because
$\ptrs{\R}{\A_0}$ contains a rule $a_n(x) \to x$.
In this case, the reduction $\rho$ has the following form,
\[
h \lrstep{*}{\A'} \stack{a_1\ldots a_n}{q} \lrstep{}{\A'} \ustack{a_1\ldots a_n}{q} = s
\]
By induction hypothesis ($ii$) applied to $h \lrstep{*}{\A'} \stack{a_1\ldots a_n}{q}$ 
there exists $h_1 \in \H(\F)$
such that $h_1 \lrstep{*}{\B} \stack{a_1}{q}$ 
and $a_1(h_1) \lrstep{a_1\ldots a_n}{\R} a_n(h)$.
Therefore, 
$a_1(h_1) \lrstep{\*}{\B} q$ and
$a_1(h_1) \lrstep{a_1\ldots a_n}{\R} h$, and~($i$) holds for $h$ and $s$.
\qed
\end{proof}
\end{RR}

\begin{corollary}
$L(\A') \subseteq \post^*_{\ptrs{\R}{\A}}(L)$
\end{corollary}
\begin{proof}
By definition of $Q^\final$, $h \in L(\A')$ iff 
$h \lrstep{*}{\A'} q \in P^\final = Q_L^\final$, 
and $P^\final \subseteq P \setminus \Pin$.
By Lemma~\ref{lem:correctness}, case $(iii)$, 
it follows that $h \in \post^*_{\ptrs{\R}{\A}}(L(\B, q)) \subseteq \post^*_{\ptrs{\R}{\A}}(L)$.
\qed
\end{proof}

\begin{RR}
\subsubsection{Completeness.}
\end{RR}

\begin{lemma}[Completeness]
For all $h \in \H(\F)$ and $s \in Q$, 
if $h \lrstep{*}{\A_0} s$ and $h \lrstep{*}{\R} h'$, 
then $h' \lrstep{*}{\A'} s$.
\end{lemma}
\begin{RR}
\begin{proof}
\end{RR}
The proof is by induction on the length of the rewrite sequence $h \lrstep{*}{\R} h'$
\begin{ABS}
(see~\cite{versionlongue}).
\end{ABS}
\begin{RR}

\noindent
If the length is 0, the result is immediate.

\noindent
Otherwise, we analyze the last rewrite step.
More precisely, assume that the rewrite step has the following form
\[
h \lrstep{*}{\R} C[a_n(h_n)] \lrstep{}{\R} C[r\sigma] = h'
\]
for some context $C[\;]$,
where the last step applies one rewrite rule 
$\rho = a_n(x) \to r$ 
and the substitution $\sigma$ associates $x$ to $h_n \in \H(\F)$.
By induction hypothesis, $C[a_n(h_n)] \lrstep{*}{\A'} s$.
This latter reduction can be decomposed as follows, 
modulo permutation of transitions,
\[
C[a_n(h_n)]_p \lrstep{*}{\A'} C[a_n(s_1\ldots s_m)]
\lrstep{*}{\A'} C[a_n(\stack{a_1\ldots a_n}{q})]
\lrstep{}{\A_0} C[\ustack{a_1\ldots a_n}{q}]
\lrstep{*}{\A'} s  
\]
where $s_1,\ldots, s_m \in Q$ and 
$a_1,\ldots, a_{n-1} \in \F$.
We show, with a case analysis over $\rho$, that 
$r \sigma \lrstep{*}{\A'} \ustack{a_1\ldots a_n}{q}$,
which implies that 
$h' = C[r \sigma] \lrstep{*}{\A'} C[\ustack{a_1\ldots a_n}{q}] \lrstep{*}{\A'} s$.

\paragraph{Case $(\REN)$:} 
$\rho = a_n(x) \to b(x)$.
In this case, 
two transitions have been added to $\A'$:
$\stack{a_1\ldots a_n}{q} \to \stack{a_1\ldots a_n b}{q}$ and
$\ustack{a_1\ldots a_n b}{q} \to \ustack{a_1\ldots a_n}{q}$.
Hence we have,
\[ r \sigma = b(h_n) 
\lrstep{*}{\A'} b(\stack{a_1\ldots a_n}{q})
\lrstep{}{\A'} b(\stack{a_1\ldots a_n b}{q})
\lrstep{}{\A_0} \ustack{a_1\ldots a_n b}{q}
\lrstep{}{\A'} \ustack{a_1\ldots a_n}{q}
\]

\paragraph{Case $(\CHILD)$:} 
$\rho = a_n(x) \to a_n(u\, x \, v)$, 
with $u, v \in Q_\A^*$.
In this case, the following transition has been added to $\A'$:
$u \, \stack{a_1\ldots a_n}{q} \, v \to \stack{a_1\ldots a_n}{q}$,
and $r \sigma = a_n(h_1\, h_n\, h_2)$ 
where $h_1 \lrstep{*}{\A} u$ and $h_2 \lrstep{*}{\A} v$.
Hence, using Lemma~\ref{cor:A0} for the first steps,

\[
r \sigma = a_n(h_1\, h_n\, h_2) 
\lrstep{*}{\A_0} a_n(u\, h_n\, v)
\lrstep{*}{\A'} a_n(u\, \stack{a_1\ldots a_n}{q}\, v)
\lrstep{}{\A'} a_n(\stack{a_1\ldots a_n}{q})
\lrstep{}{\A_0} \ustack{a_1\ldots a_n}{q}
\]

\paragraph{Case $(\SIB)$:} 
$\rho = a_n(x) \to u\, a_n(x) \, v$, 
with $u, v \in Q_\A^*$.
In this case, the following transition has been added to $\A'$:
$u \, \ustack{a_1\ldots a_n}{q} \, v \to \ustack{a_1\ldots a_n}{q}$, 
and $r \sigma = h_1\, a_n(h_n)\, h_2$ 
where $h_1 \lrstep{*}{\A} u$ and $h_2 \lrstep{*}{\A} v$.
Hence it holds that (using Lemma~\ref{cor:A0} for the first steps)
\[
r \sigma = h_1\, a_n(h_n)\, h_2 
\lrstep{*}{\A_0} u\, a_n(h_n)\, v
\lrstep{*}{\A'} u\, a_n(\stack{a_1\ldots a_n}{q})\, v
\lrstep{}{\A_0} u\, \ustack{a_1\ldots a_n}{q}\, v
\lrstep{}{\A'} \ustack{a_1\ldots a_n}{q}
\]

\paragraph{Case $(\PAR)$:} 
$\rho = a_n(x) \to b\bigl(a_n(x)\bigr)$.
In this case, the following vertical transitions have been added to $\A'$:
$b(\ustack{a_1\ldots a_n}{q}) \to \ustack{a_1\ldots a_n}{q}$, 
and we have:
\[
r \sigma =  b\bigl(a_n(h_n)\bigr)
\lrstep{*}{\A'} b\bigl(a_n(\stack{a_1\ldots a_n}{q})\bigr)
\lrstep{}{\A_0} b(\ustack{a_1\ldots a_n}{q})
\lrstep{}{\A'} \ustack{a_1\ldots a_n}{q}
\]

\paragraph{Case $(\RPL)$:}
$\rho = a_n(x) \to u$, with $u \in Q_\A^*$.
In this case, the following transition has been added to $\A'$:
$u \to \ustack{a_1\ldots a_n}{q}$.
It holds that $r \sigma \lrstep{*}{\A} u$, hence
\(
r \sigma \lrstep{*}{\A_0} u \lrstep{}{\A'} \ustack{a_1\ldots a_n}{q}\), 
using Lemma~\ref{cor:A0} for the first steps.

\paragraph{Case $(\DEL)$:}
$\rho = a_n(x) \to x$.
In this case, the following transition has been added to $\A'$:
$\stack{a_1\ldots a_n}{q} \to \ustack{a_1\ldots a_n}{q}$, and
\(
r \sigma = h_n
\lrstep{*}{\A'} \stack{a_1\ldots a_n}{q}
\lrstep{}{\A'} \ustack{a_1\ldots a_n}{q}
\).
\qed
\end{proof}
\end{RR}
As another consequence of the result of \cite{JR-rta2008} on the rewrite closure of \HA languages
under monadic \TRS, 
the backward closure of a \HA language under an \uPTRS is \HA.

\medskip
The rules of type $(\REN)$, $(\SIB)$, $(\PAR)$ and $(\RPL)$ can be easily simulated
by the \TRS of Theorem~\ref{th:post*-update}.
In particular, the parameters' semantics can be simulated 
using ground rewrite rules 
(with such rules, a symbol can generate a \HA language).
The rules $(\CHILD)$ are not 1-childvar and
the rules $(\DEL)$ is not inverse-monadic.

Example~\ref{ex:gen} shows the problems that can arise 
when combining in one single rewrite rule two rules of the form $(\SIB)$ and $(\CHILD)$, 
forcing synchronization of two updates.
Note that the rule $a(x) \to c\,a(e\,x\,g)\,d$ of this example can be simulated by the 2 rules
\( a(x) \to c\,a'(x)\,d \) and \( a'(x) \to a(e\,x\,g) \).
The former rule is of the type of Theorem~\ref{th:post*-update}
(it combines types $(\SIB)$ and $(\REN)$).
The latter (which is not 1-varchild) combines types $(\CHILD)$  and $(\REN)$.
This shows that such combinations can also lead to the behavior exposed in Example~\ref{ex:gen}.

\section*{Future Works}
As for future works on \sCFHA languages  several directions deserve to 
be followed. 
A first direction might be to  derive  pumping properties for these classes of languages.

A second direction would be to look for an analogous of Parikh characterization 
for the number of different symbols occurring in the hedges of given \sCFHA languages.
One may define and study HRS with counting constraints on horizontal and vertical paths. 

Finally, it would be is worth investigating the parallel rewriting of \cite{Delzanno12gandalf}, 
on all $a$-positions,
since it is  closer to the semantics of XQUF,
and get an analogous of Theorem~\ref{th:post*-update} for the parallel rewrite closure. 


\bibliographystyle{splncs03} 

\bibliography{HA} 

\end{document}